\documentclass[11pt, letterpaper]{article}
\usepackage[utf8]{inputenc}

\pdfoutput=1  

\usepackage{amsmath}
\usepackage{amssymb}
\usepackage{amsthm}
\usepackage{amsfonts}
\usepackage{graphicx}
\usepackage{float}
\usepackage{color}
\usepackage{url}
\usepackage[noend]{algorithm2e}
\usepackage{multirow}

\usepackage[margin=1.3in]{geometry}

\makeatletter
\newtheorem*{rep@theorem}{\rep@title}
\newcommand{\newreptheorem}[2]{%
	\newenvironment{rep#1}[1]{%
		\def\rep@title{#2 \ref{##1}}%
		\begin{rep@theorem}}%
		{\end{rep@theorem}}}
\makeatother

\newcommand{\R}{\mathbb{R}}

\newcommand{\F}{\mathcal{F}}

\newcommand{\B}{\mathcal{B}}

\renewcommand{\S}{(S_1,\ldots,S_k)}

\newcommand{\SMP}{\textsc{Sub-$k$-MP}}
\newcommand{\MSMP}{\textsc{Mon-Sub-$k$-MP}}
\newcommand{\SSMP}{\textsc{Sym-Sub-$k$-MP}}
\newcommand{\KCUT}{\textsc{Graph-$k$-Cut}}
\newcommand{\HKCUT}{\textsc{Hypergraph-$k$-Cut}}
\newcommand{\kway}{\textsc{$k$-way}}
\newcommand{\MA}{MA-Min($\F$)}
\newcommand{\MAk}{\textsc{$k$-way} MA-Min($\F$)}

\newtheorem{theorem}{Theorem}
\newtheorem{corollary}{Corollary}
\newtheorem{claim}{Claim}
\newreptheorem{theorem}{Theorem}
\newreptheorem{corollary}{Corollary}

\title{New Approximations and Hardness Results for Submodular Partitioning Problems}
\author{Richard Santiago \footnote{\texttt{rtorres@ethz.ch}, ETH Z\"{u}rich, Switzerland.}}
\date{\vspace{-5ex}}

\begin{document}
\maketitle
	
\begin{abstract}

We consider the following class of submodular k-multiway partitioning problems: (Sub-$k$-MP) $\min \sum_{i=1}^k f(S_i): S_1 \uplus S_2 \uplus \cdots \uplus S_k = V \mbox{ and } S_i \neq \emptyset \mbox{ for all }i\in [k]$. Here $f$ is a non-negative submodular function, and $\uplus$ denotes the union of disjoint sets.
Hence the goal is to partition $V$ into $k$ non-empty sets $S_1,S_2,\ldots,S_k$ such that $\sum_{i=1}^k f(S_i)$ is minimized. 
These problems were introduced by  Zhao et al. partly motivated by applications to network reliability analysis, VLSI design, and hypergraph cut, and other partitioning problems.

In this work we revisit this class of problems and shed some light on their hardness of approximation in the value oracle model. We provide new unconditional hardness results for Sub-$k$-MP in the special settings where the function $f$ is either monotone or symmetric. 
We then extend Sub-$k$-MP to a larger class of partitioning problems, where the functions $f_i(S_i)$ can be different, and there is a more general partitioning constraint $ S_1 \uplus S_2 \uplus \cdots \uplus S_k \in \mathcal{F}$ for some family $\mathcal{F} \subseteq 2^V$ of feasible sets. We provide a black box reduction that allows us to leverage several existing results from the literature; leading to new approximations for this class of problems. 

\end{abstract}



\section{Introduction}
\label{sec:intro-nonempty}

Submodularity is a property of set functions equivalent to the notion of diminishing returns. We say that a set function $f:2^V \to \R$ is \emph{submodular} if for any two sets $A\subseteq B \subseteq V$ and an element $v \notin B$, the corresponding marginal gains satisfy $f(A \cup \{v\}) -f(A) \geq f(B \cup \{v\}) -f(B)$. 
Submodular functions are a classical object in combinatorial optimization and operations research \cite{lovasz1983submodular}. They arise naturally in many contexts such as set covering problems, cuts in graphs, and facility location problems.
In recent years, they have found a wide range of applications in different computer science areas.

Since a submodular function is defined over an exponentially large domain, as is typical in the field
we assume access to a \emph{value oracle} that returns $f(S)$ for a given set $S$.
A great variety of submodular maximization and minimization problems under a wide range of constraints have been considered in the literature. In this work, we are primarily interested in the following class of submodular partitioning problems:

\textsc{Submodular $k$-Multiway Partitioning (\SMP{})}: Given  a non-negative submodular function $f:2^V \to \R_+$ over a ground set $V$, the goal is to partition $V$ into $k$ non-empty sets $S_1,S_2,\ldots,S_k$ such that $\sum_{i=1}^k f(S_i)$ is minimized. 
That is, $$ (\mbox{\SMP{}}) \quad \min \sum_{i=1}^k f(S_i): S_1 \uplus S_2 \uplus \cdots \uplus S_k = V \mbox{ and } S_i \neq \emptyset \mbox{ for all }i\in [k],$$
where we use $\uplus$ to denote the union of disjoint sets.

Special important cases occur when in addition the function $f$ is either monotone or symmetric. We refer to those as \MSMP{} and \SSMP{} respectively.
Recall that a set function $f$ is monotone if $f(A) \leq f(B)$ whenever $A \subseteq B \subseteq V$, and symmetric if $f(S) = f(V \setminus S)$ for any $S \subseteq V$.

In the absence of the non-emptyness constraints $S_i \neq \emptyset$, the problem is trivial since the partition $(V,\emptyset,\ldots,\emptyset)$ is always optimal by submodularity. However, although at first glance
the non-emptyness constraints may seem inconspicuous, they lead to interesting models and questions
in terms of tractability. We discuss this in more detail next.

These problems were introduced by Zhao, Nagamochi, and Ibaraki in \cite{zhao2005greedy} partly motivated by applications to hypergraph cut and partition problems.
They mention how \SMP{} arises naturally in settings like network reliability analysis (\cite{zhao2002thesis})
and VLSI design (\cite{chopra1999note}).
They also discuss how this class captures several important problems as special cases. 
For instance, the well-studied \KCUT{} problem in graphs where the goal is to remove a subset of edges of minimum weight such that the remaining graph has at least $k$ connected components. This problem is a special case of \SSMP{}, where $f$ corresponds to a cut function in a graph and hence it is symmetric and submodular. Another example is the more general \HKCUT{} problem on hypergraphs, where the goal is to remove a subset of hyperedges of minimum weight such that the remaining hypergraph has at least $k$ connected components. This problem is a special case of \SMP{} (see \cite{zhao2005greedy,chekuri2011approximation} for further details).

The above class of submodular partitioning problems, however, is not as well understood. 
No hardness of approximation for these problems seems to be known under the standard $P \neq NP$ assumption.
In fact, it is not known whether these problems are in P for fixed values of $k>4$ (even in the simpler monotone and symmetric cases). We discuss this in more detail in  Section \ref{sec:related-work}.

One goal of this work is to revisit these problems and shed some light onto their hardness of approximation in the value oracle model.
These hardness results are, thus, information theoretic. That is, limits on the approximability of a problem when only polynomially many queries to the value oracle are allowed. We provide new hardness results  for \SSMP{} and \MSMP{}. 
 
A second goal is to extend \SMP{} to a more general class of problems and initiate the study of its tractability. This seems natural given that \SMP{} already captures fundamental problems such as \KCUT{} and \HKCUT{}, and in addition, its complexity is not as well understood. 
We consider the class of problems given by
\begin{equation*}
\mbox{\textsc{$k$-way} MA-Min($\F$) ~~~~~min~}\sum_{i=1}^k f_i(S_i): S_1 \uplus \cdots \uplus S_k \in \F \mbox{ and } S_i \neq \emptyset \mbox{ for all }i\in [k],
\end{equation*}
where the functions $f_i$ are all non-negative submodular and potentially different, 
and the family $\F \subseteq 2^V$ can be any collection of subsets of $V$.
We denote this class by \textsc{$k$-way Multi-Agent Minimization} (\textsc{$k$-way} MA-Min).
This is partially motivated by the work of Santiago and Shepherd \cite{santiago2018ma} which considers the following class of multi-agent submodular minimization problems:
\begin{equation*}
\mbox{MA-Min($\F$) ~~~~~min~}\sum_{i=1}^k f_i(S_i): S_1 \uplus \cdots \uplus S_k \in \F.
\end{equation*}

We study the connections between these multi-agent problems and their \kway{} versions.
In particular, we show that in many cases the approximation guarantees for MA-Min($\F$) can be extended to the corresponding \kway{} versions at a small additional loss.

\subsection{Related work}
\label{sec:related-work}
Zhao et al.  \cite{zhao2005greedy} show that a simple greedy splitting algorithm achieves a $(2-2/k)$-approximation for both \MSMP{} and \SSMP{} (Queyranne \cite{queyranne1999optimum} announced the same result for symmetric submodular functions), and a $(k-1)$-approximation for the more general \SMP{}. All these approximations hold for arbitrary (i.e., not necessarily fixed) values of $k$. Okumoto et al. \cite{okumoto2012divide} showed that \SMP{} is polytime solvable for $k=3$, and Gui\~{n}ez and Queyranne \cite{guinezsize} showed that the symmetric version \SSMP{} is polytime solvable for $k=4$. We next discuss in more detail the cases where $k$ is fixed (i.e., not part of the input) and when $k$ is part of the input.

\underline{For fixed values of $k$}, the \KCUT{} problem can be solved in polynomial time \cite{goldschmidt1994polynomial,karger1996new}. 
In recent work Chandrasekaran et al. \cite{chandrasekaran2019hypergraph} gave a randomized polytime algorithm for \HKCUT{}, whose complexity had remained an intriguing open problem even for fixed values of $k$. 
In subsequent work Chandrasekaran and Chekuri \cite{chandrasekaran2020hypergraph} gave a deterministic polytime algorithm.
For the more general submodular multiway partitioning problems,  
Chekuri and Ene \cite{chekuri2011approximation} gave a $(1.5 - 1/k)$-approximation for \SSMP{} and a $2$-approximation for \SMP{}. The latter was improved to $2-2/k$ by Ene et al. \cite{ene2013local}. 

\underline{When $k$ is part of the input} \KCUT{} is NP-Hard \cite{goldschmidt1994polynomial}. Hence all the above problems are also NP-Hard (see Appendix \ref{app:monotone-hardness} for details about the monotone case). Moreover, the symmetric and general version are also APX-Hard since they generalize \KCUT{}. We note that \KCUT{} was claimed to be APX-Hard by Papadimitriou (see \cite{saran1995finding}), although a formal proof never appeared in the literature until the recent work of Manurangsi \cite{manurangsi2018inapproximability}. The latter gave conditional hardness by showing that assuming the Small Set Expansion Hypothesis, it is NP-hard to approximate \KCUT{} to within a $2-\epsilon$ factor of the optimum for every constant $\epsilon>0$.
Chekuri and Li \cite{chekuri2015note} give a simple reduction showing that an $\alpha$-approximation for \HKCUT{} implies an $O(\alpha^2)$-approximation for \textsc{Densest-$k$-Subgraph}.
This gives conditional hardness of approximation for \HKCUT{} since the best known approximation for \textsc{Densest-$k$-Subgraph} is $O(n^{1/4+\epsilon})$ \cite{bhaskara2010detecting}, and Manurangsi \cite{manurangsi2017almost} shows that assuming the Exponential Time Hypothesis there is no polynomial-time algorithm with an approximation factor of $n^{1/(\log \log n)^c}$ for some constant $c>0$. 
The \textsc{Densest-$k$-Subgraph} problem is believed to not admit an efficient constant factor
approximation assuming $P \neq NP$.
Since \SMP{} generalizes \HKCUT{}, the above gives conditional hardness on \SMP{}. 


The hardness of approximation for the class of  submodular multiway partitioning problems, however, is not as well understood.
No hardness of approximation for these problems seems to be known under the $P \neq NP$ assumption or under the value oracle model.
In fact, it is not even known whether \MSMP{}, \SSMP{}, or \SMP{}, are in P for fixed low values of $k$ (read $k>3$ for the monotone and general versions, and $k>4$ for the symmetric version).

For the class of multi-agent minimization problems MA-Min($\F$), the special case where $\F = \{V\}$ is known as Minimum Submodular Cost Allocation and has been previously studied \cite{hayrapetyan2005network,svitkina2010facility,ene2014hardness,chekuri2011submodular}. The works of Goel et al. \cite{goel2009approximability} and Santiago and Shepherd \cite{santiago2018ma} studied these problems under more general families. 
For a comprehensive review of multi-agent submodular optimization problems see \cite{santiago2019phdthesis}.


We are not aware of previous work for the \textsc{$k$-way} MA-Min($\F$) class of problems besides the special case of \SMP{}. That is, the case with $\F=\{V\}$ and $f_i =f$ for all $i$.

\subsection{Our contributions}

The contributions of this work are three-fold: new hardness results, a black box reduction, and new applications. We discuss each of these three blocks next.

In this work we initiate the study of the hardness of approximation in the value oracle model for different variants of  \SMP{}. 
We provide the first unconditional hardness of approximation results for \SSMP{} and \MSMP{} in the value oracle model.
For the latter problem we are not aware of any previous (even conditional) hardness 
result.
For \SSMP{} our bound matches the (conditional) inapproximability factor of $2-\epsilon$ from the work of \cite{manurangsi2018inapproximability}.
See Section~\ref{sec:hardness} for proof details and further discussion. 

\begin{theorem}
	\label{thm:SSMP}
	Given any $\epsilon > 0$, any algorithm achieving a $(2 - \epsilon)$-approximation for the \SSMP{} problem when $k$ is part of the input, requires exponentially many queries in the value oracle model.
\end{theorem}

\begin{theorem}
	\label{thm:MSMP}
	Given any $\epsilon > 0$, any algorithm achieving a $(\frac{4}{3} - \epsilon)$-approximation for the \MSMP{} problem when $k$ is part of the input, requires exponentially many queries in the value oracle model.
\end{theorem}

Our main algorithmic result is a black box procedure which, at a small additional loss, turns a solution for the multi-agent problem (i.e., MA-Min($\F$)) into a solution for the \kway{} version (i.e., \kway{} MA-Min($\F$)). We do this in the case where the objective functions $f_i$ are non-negative and monotone, and for families $\F$ that are upwards closed (i.e., if $S\in \F$ and $T \supseteq S$ then $T \in \F$). The latter is a mild assumption given that the functions are monotone. Our guarantees are tight up to a small constant additive term. 

\begin{theorem}
	\label{thm:+1gap}
	Let $\F$ be an upwards closed family.
	Then an $\alpha(n,k)$-approximation for monotone MA-Min($\F$) implies an $(\alpha(n,k)+1)$-approximation for monotone \textsc{$k$-way} MA-Min($\F$). In addition, there are instances where achieving an $(\alpha(n,k)+\frac{1}{3}-\epsilon)$-approximation requires exponentially many queries in the value oracle model for any $\epsilon >0$.
\end{theorem}

We remark that improving the above additive term would lead to an improvement of the best current approximation factor for \MSMP{} (i.e., the setting where $\F=\{V\}$ and $f_i=f$ for all $i$). The multi-agent version of this problem: $\min \sum_{i=1}^k f(S_i): S_1  \uplus \cdots \uplus S_k = V$, has a trivial $1$-approximation by taking the partition $V,\emptyset,\ldots,\emptyset$, and hence by Theorem~\ref{thm:+1gap} we obtain a $2$-approximation for the corresponding \kway{} version, i.e., for \MSMP{}. This matches asymptotically the best known approximation of $2-2/k$ for this problem given in \cite{zhao2005greedy}. It is not known however whether this is tight.

The above black box result leads to interesting applications (see Section~\ref{sec:min-nonempty} for full details). For instance, the problem \kway{} MA-Min($\F$) with monotone functions $f_i$ and where $\F$ corresponds to the family of vertex covers of a graph, admits a tight $O(\log n)$-approximation. Moreover, in the case where all the functions $f_i$ are the same, this becomes a $3$-approximation. To the best of our knowledge this is the current best approximation for this problem. 
The special case $k=1$ corresponds to the submodular vertex cover problem studied in \cite{goel2009approximability}, where a hardness of $2-\epsilon$ is shown in the value oracle model.

\begin{corollary}
	\label{cor:3-approx-vertex-cover}
	There is a $3$-approximation algorithm for the problem $\min \sum_{i \in [k]} f(S_i): S_1 \uplus \cdots \uplus S_k \in \F$ and $S_i \neq \emptyset$ for all $i \in [k]$, where $f$ is non-negative monotone submodular and $\F$ is the family of vertex covers of a graph. 
\end{corollary}

\begin{table}[t]
\centering
\resizebox{\textwidth}{!}{
\begin{tabular}{c|c|c|c|c|c|c|}
\cline{2-7}
& \multicolumn{2}{ c| }{ \multirow{2}{*}{\MSMP} } &  \multicolumn{2}{c|}{ \multirow{2}{*}{\SSMP} } &  \multicolumn{2}{ c| }{ \MSMP }  \\ 
&  \multicolumn{2}{c|}{}  &  \multicolumn{2}{c|}{} &   \multicolumn{2}{ c| }{ over vertex covers }   \\ \cline{2-7}
& Approx & Hardness & Approx & Hardness & Approx & Hardness  \\ \cline{1-7}
\multicolumn{1}{ |c| }{Known} & $2-2/k$ \cite{zhao2005greedy} & - & $2-2/k$ \cite{zhao2005greedy} & $2-\epsilon$ (conditional) \cite{manurangsi2018inapproximability}  & - & $2-\epsilon$ \cite{goel2009approximability} \\ \cline{1-7}
\multicolumn{1}{ |c| }{This paper} & $2$ (faster) & $4/3 - \epsilon$ & - & $2-\epsilon$ (unconditional)   & $3$  & -  \\ \cline{1-7} 
\end{tabular}
}
\caption{Comparison of some of the results in this paper with previous work.}
\label{table:results}
\end{table}

Another direct consequence of Theorem~\ref{thm:+1gap} is providing a very simple $2$-approximation for \MSMP{}.
The argument in fact shows that one very specific partition achieves the desired bound. 
In addition to simple, the procedure to build such a partition is also fast. 
Indeed, the running time of the $(2-2/k)$-approximation algorithm provided by Zhao et al. \cite{zhao2005greedy} is $kn^3EO$, where $EO$ denotes the time that a call to the value oracle takes. On the other hand, the running time of this procedure is $O(n \, EO + n\log n)$ and hence almost linear. 

\begin{corollary}
	\label{cor:faster-algo}
	There is a $2$-approximation algorithm for \MSMP{} running in time $O(n \, EO + n\log n)$, where $EO$ denotes the time that a call to the value oracle takes.
\end{corollary}

We summarize some of our results and compare them with previous work in Table~\ref{table:results}.

\section{Hardness results in the value oracle model}
\label{sec:hardness}

In this section we provide the first unconditional hardness of approximation results in the value oracle model for \SSMP{} and \MSMP{}. 
In addition, for \MSMP{} we are not aware of any (even conditional) hardness of approximation result previous to this work.

Our results are based on the technique of building two functions that are hard to distinguish with high probability for any (even randomized) algorithm. This was first used in the work of Goemans et al. \cite{goemans2009approximating}, and has since then been used in several subsequent works \cite{feige2011maximizing,svitkina2011submodular,goel2009approximability,iwata2009submodular,santiago2016multivariate}.

\subsection{A $2$-factor inapproximability oracle hardness for \SSMP{}}
\label{sec:sym}

The current best known (conditional) hardness of approximation for \SSMP{} follows from the result of Manurangsi \cite{manurangsi2018inapproximability}, where it is shown that assuming the Small Set Expansion Hypothesis, it is NP-hard to approximate \KCUT{} to within a $2-\epsilon$ factor of the optimum for every constant $\epsilon>0$. Since \SSMP{} generalizes \KCUT{}, the same conditional hardness of approximation automatically applies.
In this section we prove an unconditional lower bound hardness for \SSMP{} in the value oracle model (Theorem \ref{thm:SSMP}). To the best of our knowledge this is the first result of this kind for this problem. 

To prove the desired result, we first build two indistinguishable functions as follows.
Let $|V|=n$ be an even number, $R$ be a random set of size $\frac{n}{2}$, and $\bar{R}$ denote its complement. Define parameters $\epsilon^2 = \frac{1}{n} \omega(\ln n)$ and $\beta = \frac{n}{4}(1+\epsilon)$, such that $\beta$ is an integer. Consider the functions
\begin{equation*}
f_1(S) = \min\big\{|S|, \frac{n}{2}\big\} - \frac{|S|}{2} \;\; \mbox{ and } \;\; f_2(S) = \min\big\{|S|, \frac{n}{2}, \beta + |S \cap R|, \beta + |S \cap \bar{R}| \big\} - \frac{|S|}{2} .
\end{equation*}
These functions were already used in the work of Svitkina and Fleischer \cite{svitkina2011submodular} to prove polynomial hardness of approximation for the submodular sparsest cut and submodular balanced cut problems. They show that the above two functions are non-negative symmetric submodular and hard to distinguish. That is, any (even randomized) algorithm that makes a polynomial number of oracle queries has probability at most $n^{-\omega(1)}$ of distinguishing the functions $f_1$ and $f_2$. 

We use this to show hardness of approximation for \SSMP{} as follows.

\begin{claim}
	\label{claim:sym}
	Consider the \SSMP{} problem with $k=\frac{n}{2}+1$ and inputs $f_1$ and $f_2$. Then, if the input is $f_1$ any feasible solution has objective value at least $\frac{n}{2}$, while if the input is $f_2$ then the optimal value is at most $\frac{n}{4}(1+\epsilon)$.
\end{claim}
\begin{proof}
	Since we have $n$ elements that must be split into $\frac{n}{2}+1$ non-empty sets, no more than $\frac{n}{2}$ items can be assigned to any given set. That is, for any feasible solution $S_1,S_2,\ldots,S_k$ we must have that $|S_i|\leq \frac{n}{2}$. It then follows that when the input is $f_1$, any feasible solution $S_1,S_2,\ldots,S_k$ has objective value exactly 	
	$\sum_{i=1}^k f_1(S_i) = \sum_{i=1}^k \big(|S_i| - \frac{|S_i|}{2} \big) = \frac{n}{2}$.
	
	On the other hand, when the input is $f_2$, a feasible solution is given by taking $S_1 = \bar{R}, S_2 = \{r_1\}, S_3 = \{r_2\},\ldots, S_k = \{r_{n/2}\}$, where $R = \{r_1,r_2,\ldots,r_{n/2}\}$. This has objective value
	$$\sum_{i=1}^k f_2(S_i) = f_2(\bar{R}) + \sum_{i=2}^k f_2(r_{i-1}) = \Big(\beta - \frac{n}{4}\Big) + \sum_{i=2}^k \Big(1-\frac{1}{2}\Big)
	= \beta - \frac{n}{4} +  \frac{n}{4} = \beta = \frac{n}{4}\big(1+\epsilon\big).$$
\end{proof}

From the above result it follows that the gap between the optimal solutions for \SSMP{} when the inputs are $f_1$ and $f_2$ is at least
\begin{equation*}
\frac{OPT_1}{OPT_2} \geq \frac{ \frac{n}{2} }{\frac{n}{4} (1+\epsilon)} = \frac{2}{1+\epsilon}.
\end{equation*}
Since $\epsilon=o(1)$ this gap can be arbitrarily close to $2$ for large values of $n$. Given that $f_1$ and $f_2$ are hard to distinguish, this now leads to Theorem \ref{thm:SSMP}.

\begin{proof}[Proof [Theorem \ref{thm:SSMP}]]
	Assume there is an algorithm that makes polynomially many queries to the value oracle and that achieves a $(2-\delta)$-approximation for \SSMP{} for some constant $\delta > 0$. 
	Let the functions $f_1$ and $f_2$ be as defined above, and choose $n$ and the parameter $\epsilon(n)$ so that 
	$(1+\epsilon(n))(2-\delta)<2$, i.e., so that $\epsilon(n) < \delta/(2-\delta)$. Since $\epsilon(n)=o(1)$ and $\delta$ is a constant, this is always possible.
	
	Consider the output of the algorithm when the input is $f_2$. By Claim \ref{claim:sym} in this case the optimal solution is at most $\frac{n}{4} (1+\epsilon)$, and hence the algorithm finds a feasible solution $(S_1,S_2,\ldots,S_k)$ such that $\sum_{i\in [k]} f_2(S_i) \leq (2-\delta) (1+\epsilon) \frac{n}{4} < \frac{n}{2}$, where the last inequality follows from the choice of $\epsilon$. However, there is no feasible solution $(S'_1,S'_2,\ldots,S'_k)$ such that $\sum_{i\in [k]} f_1(S'_i) < \frac{n}{2}$, since by Claim \ref{claim:sym} any feasible solution for $f_1$ has value at least $n/2$. That means that if the input is the function $f_1$ the algorithm outputs a different answer, thus distinguishing between $f_1$ and $f_2$. A contradiction. 
\end{proof}

\subsection{A $4/3$-factor inapproximability oracle hardness for \MSMP{}}
\label{sec:monotone-nonempty}

In this section we prove an unconditional lower bound hardness of approximation for \MSMP{} in the value oracle model (Theorem \ref{thm:MSMP}).
To the best of our knowledge, this is the first hardness of approximation result (either conditional or unconditional) for \MSMP{}. As discussed in Appendix \ref{app:monotone-hardness}, the conditional hardness of approximation for \KCUT{} does not extend to \MSMP{}, since the objective function in that case must take negative values.

The argument is  similar to the one from Section \ref{sec:sym}. We consider the two functions
\begin{equation*}
f_3(S) = \min\big\{|S|, \frac{n}{2}\big\}  \;\; \mbox{ and } \;\; f_4(S) = \min\big\{|S|,\frac{n}{2}, \beta + |S \cap R|, \beta + |S \cap \bar{R}| \big\},
\end{equation*}
where all the parameters are as defined in Section \ref{sec:sym}. Note that $f_3 = f_1 + g$ and $f_4 = f_2 + g$, where $g(S) = |S|/2$. Since both $f_1$ and $f_2$ are submodular, and $g$ is modular, it follows that $f_3$ and $f_4$ are also submodular. Moreover, it is straightforward to check that both $f_3$ and $f_4$ are also non-negative and monotone.

Since $f_1$ and $f_2$ are hard to distinguish, and $f_3(S) \neq f_4(S)$ if and only if $f_1(S) \neq f_2(S)$, it follows that $f_3$ and $f_4$ are also hard to distinguish.

The following result shows the gap between the optimal solutions of the corresponding problems.

\begin{claim}
	\label{claim:monot}
	Consider the \MSMP{} problem with $k=\frac{n}{2}+1$ and inputs $f_3$ and $f_4$. Then, if the input is $f_3$ any feasible solution has objective value at least $n$, while if the input is $f_4$ then the optimal value is at most $\frac{3+\epsilon}{4}n$.
\end{claim}
\begin{proof}
	The argument is very similar to that of Claim \ref{claim:sym}.
	Since any feasible solution $S_1,S_2,\ldots,S_k$ must satisfy $|S_i|\leq \frac{n}{2}$ for all $i$, it then follows that when the input is $f_3$, any feasible solution $S_1,S_2,\ldots,S_k$ has objective value exactly 	
	$\sum_{i=1}^k f_3(S_i) = \sum_{i=1}^k |S_i| = n$.
	
	On the other hand, when the input is $f_4$, a feasible solution is given by $S_1 = \bar{R}, S_2 = \{r_1\}, S_3 = \{r_2\},\ldots, S_k = \{r_{n/2}\}$, where $R = \{r_1,r_2,\ldots,r_{n/2}\}$. This has objective value
	$\sum_{i=1}^k f_4(S_i) = \frac{n}{4}(1+\epsilon) + \frac{n}{2} = \frac{3+\epsilon}{4}n$.
\end{proof}

It follows that the gap between the optimal solutions for \MSMP{} when the inputs are $f_3$ and $f_4$ is at least
\begin{equation*}
\frac{OPT_3}{OPT_4} \geq \frac{4}{3+\epsilon}.
\end{equation*}
Since $\epsilon=o(1)$ this gap can be arbitrarily close to $4/3$ for large values of $n$. Given that $f_3$ and $f_4$ are hard to distinguish, this now leads to Theorem \ref{thm:MSMP}.

\begin{proof}[Proof [Theorem \ref{thm:MSMP}]]
	Assume there is an algorithm that makes polynomially many queries to the value oracle and that achieves a $(4/3-\delta)$-approximation for \MSMP{} for some constant $\delta > 0$. 
	Let the functions $f_3$ and $f_4$ be as defined above, and choose $n$ and the parameter $\epsilon(n)$ so that 
	$(\frac{3+\epsilon}{4})(\frac{4}{3}-\delta)<1$, i.e., so that $\epsilon < \frac{9 \delta}{4-3\delta}$. Since $\epsilon(n)=o(1)$ and $\delta$ is a constant, this can always be done.
	
	Consider the output of the algorithm when the input is $f_4$. By Claim \ref{claim:monot} in this case the optimal solution is at most $\frac{3+\epsilon}{4}n$, and hence the algorithm finds a feasible solution $(S_1,S_2,\ldots,S_k)$ such that $\sum_{i\in [k]} f_4(S_i) \leq (\frac{4}{3}-\delta) \frac{3+\epsilon}{4}n < n$, where the last inequality follows from the choice of $\epsilon$. However, there is no feasible solution $(S'_1,S'_2,\ldots,S'_k)$ such that $\sum_{i\in [k]} f_3(S'_i) < n$, since by Claim \ref{claim:monot} any feasible solution for $f_3$ has value at least $n$. That means that if the input is the function $f_3$ the algorithm outputs a different answer, thus distinguishing between $f_3$ and $f_4$. A contradiction. 
\end{proof}

\section{From multi-agent minimization to the \textsc{$k$-way} versions}
\label{sec:min-nonempty}

In this section we show that if the functions $f_i$ are monotone, then a feasible solution to the \MA{} problem can be turned into a feasible solution to the corresponding \textsc{$k$-way} version (i.e., \MAk{}) at almost no additional loss. Moreover, our argument is completely black box with respect to how the approximation for the MA-Min instance is obtained (i.e., it could be via a greedy algorithm, a continuous relaxation, or any other kind of approach). We show the following.

\begin{reptheorem}{thm:+1gap}
	Let $\F$ be an upwards closed family.
	Then an $\alpha(n,k)$-approximation for monotone MA-Min($\F$) implies an $(\alpha(n,k)+1)$-approximation for monotone \textsc{$k$-way} MA-Min($\F$). In addition, there are instances where achieving an $(\alpha(n,k)+\frac{1}{3}-\epsilon)$-approximation requires exponentially many queries in the value oracle model for any $\epsilon >0$.
\end{reptheorem}
\begin{proof}
	Denote by $OPT$ the value of the optimal solution to the \textsc{$k$-way} problem and by $\overline{OPT}$ the value of the optimal solution to \MA{}. Then it is clear that $\overline{OPT} \leq OPT$ since any feasible solution for the \textsc{$k$-way} version is also feasible for \MA{}. 
	
	Let $G=([k] \uplus V,E)$ denote the complete bipartite graph where the weight of an edge $(i,v)$ is given by $f_i(v)$. Let $M$ be a minimum $[k]$-saturating matching in $G$, that is a minimum cost matching such that every node in $[k]$ gets assigned at least one element. Since the edges have non-negative weights it is clear that $|M|=k$, i.e., each node $i \in [k]$ gets assigned exactly one element from $V$. Denote the edges of the matching by $M=\{(1,u_1), (2,u_2),\ldots,(k,u_k)\}$, and let $U:=\{u_1, u_2,\ldots,u_k\}$ be the elements in $V$ that $M$ is incident to.
	
	We then have that the cost of $M$ is at most $OPT$. Indeed, if $(S^*_1,S^*_2,\ldots,S^*_k)$ is an optimal solution to the \textsc{$k$-way} instance, we can remove elements from the sets $S^*_i$ arbitrarily until each of the sets consists of exactly one element. By monotonicity, removing elements can only decrease the objective value of the solution. Moreover, since now each set consists of exactly one element, this is a feasible $[k]$-saturating matching, and hence its cost is at least the cost of $M$. That is, $\sum_{i\in [k]} f_i(u_i) \leq OPT$.

	Let $(S_1,S_2,\ldots,S_k)$ be an $\alpha$-approximation for the \MA{} instance. Then we have $\sum_{i\in [k]} f_i(S_i) \leq \alpha \cdot \overline{OPT} \leq  \alpha \cdot OPT$. We combine this solution with the matching $M$ by defining new sets $S'_i := (S_i \setminus U) \uplus \{u_i\}$ for each $i\in [k]$. It is clear that this is now a feasible solution to the \textsc{$k$-way} problem since all the sets $S'_i$ are non-empty and pairwise disjoint, and their union
	$\cup_{i \in [k]} S'_i =  U \cup \big( \cup_{i \in [k]} S_i \big)$ belongs to $\F$
	since $\cup_{i \in [k]} S_i \in \F$ and $\F$ is upwards closed. 
	Moreover, the cost of the new solution is given by
	\begin{align*}
	\sum_{i\in [k] } f_i(S'_i) & = \sum_{i\in [k] } f_i(S_i \setminus U + u_i) \leq \sum_{i\in [k]} f_i(S_i \setminus U) +\sum_{i\in [k]} f_i(u_i) \leq
	\sum_{i\in [k]} f_i(S_i) +\sum_{i\in [k]} f_i(u_i)\\
	& \leq \alpha \cdot OPT + OPT = (\alpha + 1) \cdot OPT,
	\end{align*}
	where the first inequality follows from subadditivity (since the functions are non-negative submodular) and the second inequality from monotonicity. 
	
	For the inapproximability result part, consider the family $\F = \{V\}$ and the setting where all the functions $f_i$ are the same. Then the corresponding \MA{} problem has a trivial $1$-approximation, while the \kway{} version corresponds to \MSMP{}. The latter, by Theorem~\ref{thm:MSMP}, cannot be approximated in the value oracle model to a factor of $(\frac{4}{3} - \epsilon)$ for any $\epsilon >0$ without making exponentially many queries. It follows that for these instances, the MA-Min version has an exact solution while the \kway{} versions have an inapproximability lower bound of $4/3 - \epsilon = 1 + 1/3 - \epsilon$. 
	Hence, there are instances where for any $\epsilon >0$, achieving an $(\alpha(n,k)+\frac{1}{3}-\epsilon)$-approximation requires exponentially many queries in the value oracle model.
	This completes the proof.
\end{proof}

	For proving or improving the result from Theorem \ref{thm:+1gap}, one could be tempted to first compute an optimal (or approximate) solution $\S$ to the MA-Min($\F$) problem, and then find an allocation of some of the elements of $F:=\uplus_{i \in [k]} S_i$ among the agents that did not get any item. However, this approach can lead to a large additional loss, since a set $F\in \F$ could be optimal for the MA-Min problem but highly suboptimal for the \kway{} version. The following example shows this, even for the case of modular functions.
	
	Let $T \subsetneq V$ be an arbitrary set of size $2(k-1)$. Let $f_1(S) = |S|$ and $f_i(S)=w(S)$ for all $i\geq 2$ where $w:V\to \R_+$ is the modular function given by $w(v)=1+\epsilon$ for $v\notin T$ and $w(v)=M$ for $v\in T$, for some large value $M$. Moreover, let $\F=\{S:|S|\geq 2(k-1)\}$. Then a feasible (and optimal) solution to the \MA{} problem is given by the allocation $(T,\emptyset,\ldots,\emptyset)$ with objective value $f_1(T)=|T|=2(k-1)$. However, any splitting of some of the items of $T$ among the other $k-1$ agents leads to a solution of cost at least $M(k-1)+(k-1)=(M+1)(k-1)$. On the other hand, an optimal solution for the \kway{} version is given by any partition of the form $(S_1,\{v_2\},\{v_3\},\ldots,\{v_k\})$ where $S_1 \subseteq V$ is any set of $k-1$ elements, and $\{v_2,v_3,\ldots,v_k\} \subseteq V \setminus T$. This leads to a solution of cost $(k-1)+(1+\epsilon)(k-1) = (2+\epsilon)(k-1)$. Thus having a gap in terms of objective value of at least $\frac{(M+1)(k-1)}{(2+\epsilon)(k-1)}=\frac{M+1}{2+\epsilon}$.

Theorem \ref{thm:+1gap} allows us to extend several results from monotone multi-agent minimization to the \kway{} versions. We discuss some of these consequences next. An interesting application is obtained using the $O(\log n)$-approximation from \cite{svitkina2010facility}.

\begin{corollary}
	\label{cor:MSCA-nonempty}
	There is a tight $O(\log n)$-approximation for the allocation problem
	\begin{equation*}
	\min \sum_{i=1}^k f_i(S_i): S_1 \uplus S_2 \uplus \cdots \uplus S_k = V \mbox{ and } S_i \neq \emptyset \mbox{ for all }i\in [k],
	\end{equation*}
	where all the functions $f_i$ are non-negative monotone submodular.
\end{corollary}
\begin{proof}
	The approximation factor follows from Theorem \ref{thm:+1gap} and the tight $O(\log n)$-approximation (\cite{svitkina2010facility}) for the corresponding MA-Min instance. To see why this is tight assume that an (asymptotically) better approximation factor of $o(\log n)$ is possible. Then given a MA-Min instance we can reduce it to an instance of the \kway{} version by adding a set $D:=\{d_1,d_2,\ldots,d_k\}$ of $k$ dummy elements to the ground set. That is, consider an instance of the \kway{} version with $V':=V \cup D$, $\F'=\{V'\}$, and $f'_i(S):=f_i(S\cap V)$ for each $i\in[k]$ and $S \subseteq V'$. Then by assumption we have a $o(\log(n+k))$-approximation for this problem, and hence we also have the same approximation factor for the original MA-Min instance. But this contradicts the lower bound of $\Omega(\log n)$ for the MA-Min problem.
\end{proof}

More generally, we can obtain approximation guarantees for families with a bounded blocker. 
Given a family $\F$, there is an associated blocking clutter $\B(\F)$ which consists of the minimal sets $B$ such that 
$B \cap F \neq \emptyset$ for each $F \in \F$. We refer to $\B(\F)$ as the blocker of $\F$. 
We say that $\B(\F)$ is $\beta$-bounded if $|B|\leq \beta$ for all $B \in \B(\F)$.
Families such as $\F=\{V\}$ or vertex covers in a graph, are examples of families with a bounded blocker.
Indeed, the family $\F=\{V\}$ has a $1$-bounded blocker, since $\B(\F)=\{ \{v_1\},\{v_2\},\ldots,\{v_n\}  \}$.
The family $\F$ of vertex covers of a graph $G$ has a $2$-bounded blocker, since $\B(\F) = \{ \{u,v\}: (u,v) \mbox{ is an edge in } G  \}$.
Recall that a set $S \subseteq V$ is a vertex cover in a graph $G$ if every edge in $G$ is incident on a vertex in $S$.

It is shown in \cite{santiago2018ma} that families with a $\beta$-bounded blocker admit a $O(\beta \log n)$-approximation for the multi-agent monotone minimization problem. This, combined with Theorem \ref{thm:+1gap}, implies a $O(\beta \log n)$-approximation for the \kway{} versions. In particular, this leads to a tight $O(\log n)$-approximation for the \kway{} MA-Min($\F$) problem with monotone functions $f_i$ and where $\F$ corresponds to the family of vertex covers of a graph. The tightness follows from Corollary~\ref{cor:MSCA-nonempty} and the fact that vertex covers generalize the family $\F=\{V\}$.

\subsection{The special case where all the functions $f_i$ are the same}

Theorem \ref{thm:+1gap} also leads to interesting consequences in the special case where $f_i = f$ for all $i$. In that setting, it is easy to see that the single-agent and multi-agent versions are equivalent. That is,
\begin{equation}
\label{obs:sa-ma-min}
\min f(S): S \in \F \quad = \quad \min \sum_{i\in [k]} f(S_i):S_1 \uplus S_2 \uplus \cdots \uplus S_k \in \F,
\end{equation}
and moreover $F\in \F$ is an optimal solution to the single-agent problem if and only if the trivial partition $(F,\emptyset,\ldots,\emptyset)$ is an optimal solution to the multi-agent version. Again this just follows from submodularity and non-negativity since then $f(T)\leq f(S) + f(T-S)$ for any $S \subseteq T \subseteq V$. That is, partitioning the elements of a set can only increase the value of the solution.
This leads to the following result.

\begin{corollary}
	\label{cor:gap+1}
	Let $\F$ be any upwards closed family, and
	assume there is an $\alpha(n)$-approximation for the single-agent monotone minimization problem: $\min f(S): S \in \F$. 
	If $f_i = f$ for all $i$, then there is an $\alpha(n)$-approximation for monotone \MA{}, and hence
	an $(\alpha(n)+1)$-approximation for monotone \kway{} MA-Min($\F$). 
\end{corollary}
\begin{proof}
	By Equation \eqref{obs:sa-ma-min}, an $\alpha(n)$-approximation for the  single-agent monotone minimization problem implies an $\alpha(n)$-approximation for monotone \MA{} in the setting where $f_i = f$ for all $i$. Now the result for the corresponding \kway{} versions immediately follows from Theorem \ref{thm:+1gap}.
\end{proof}

We remark that by taking $\F=\{V\}$ the above corollary leads to a $2$-approximation for \MSMP{}, which matches asymptotically  the currently best known. Hence improving the plus one additive term would lead to an improvement on the approximation factor of the latter problem.

Corollary \ref{cor:gap+1} leads to new results. For instance, using the $2$-approximation from \cite{goel2009approximability,iwata2009submodular} for single-agent monotone minimization over families of vertex covers, we immediately get a $3$-approximation for the corresponding monotone \kway{} MA-Min($\F$) problem over the same type of families. This now proves Corollary~\ref{cor:3-approx-vertex-cover}. 
More generally, given the $\beta$-approximation results (\cite{iyer2014monotone,koufogiannakis2013greedy}) for minimizing a monotone submodular function over families with a $\beta$-bounded blocker, we have the following.

\begin{corollary}
	Let $\F$ be an upwards closed family with a $\beta$-bounded blocker. Then there is a $(\beta+1)$-approximation algorithm for the problem
	$\min \sum_{i \in [k]} f(S_i): S_1 \uplus \cdots \uplus S_k \in \F$  and $S_i \neq \emptyset$ for all $i \in [k]$, where $f$ is a non-negative monotone submodular function.
\end{corollary}

\subsection{A simpler and faster $2$-approximation for \MSMP}
Another direct consequence of Theorem \ref{thm:+1gap} is to provide a very simple and fast $2$-approximation for \MSMP.
We describe the procedure in Algorithm~\ref{alg:MSMP}.
The running time of the $(2-2/k)$-approximation algorithm provided by Zhao et al. \cite{zhao2005greedy} is $kn^3EO$, where $EO$ denotes the time that a call to the value oracle takes. On the other hand, the running time of our procedure is $O(n \, EO + n\log n)$ and hence almost linear.
All we need to do is first make $n$ oracle calls to evaluate $f(v)$ for each $v\in V$, and then sort the elements so that $f(v_1) \leq f(v_2) \leq \ldots \leq f(v_n)$ (which requires $O(n\log n)$ time).

\RestyleAlgo{algoruled}
\begin{algorithm}[ht]
	\footnotesize
	\KwIn{A ground set $V = \{v_1,v_2,\ldots,v_n\}$, and a set function $f:2^V \to \R$ with oracle access.}
	Sort and rename the elements so that $f(v_1) \leq f(v_2) \leq \ldots \leq f(v_n)$.\\
	$S_1 \leftarrow \{v_1\}, \; S_2 \leftarrow \{v_2\}, \; \ldots, \; S_{k-1} \leftarrow \{v_{k-1}\}, \; S_k \leftarrow V \setminus \{v_1,v_2,\ldots,v_{k-1}\}$.\\
	\KwOut{$(S_1,S_2,\ldots,S_k)$}
	\caption{\textsf{Simpler and faster algorithm for \MSMP{} }}
	\label{alg:MSMP}
\end{algorithm}

\begin{repcorollary}{cor:faster-algo}
	Algorithm~\ref{alg:MSMP} is a $2$-approximation algorithm for \MSMP{} running in time $O(n \, EO + n\log n)$, where $EO$ denotes the time that a call to the value oracle takes.
\end{repcorollary}

\section{Conclusion and open problems}

We revisited the class of Submodular $k$-Multiway Partitioning problems (\SMP{}). We proved new unconditional inapproximability results for the monotone and symmetric cases of \SMP{} in the value oracle model.

We introduced and explored a new class of submodular partitioning problems which generalizes \SMP{}. 
We showed that several results from multi-agent submodular minimization can be extended to their \kway{} counterparts at a small additional loss. Thus obtaining several new results for this class of problems. 

Many interesting open questions remain, perhaps the most important being about the approximation hardness of \SMP{}. 
It remains completely open whether these problems are polytime solvable for fixed values of $k>4$.
In addition, given the conditional hardness of approximation for \HKCUT{} 
based on \textsc{Densest-$k$-Subgraph}, we believe it may be possible to prove strong unconditional hardness of approximation results for \SMP{} in the value oracle model when $k$ is part of the input. 

It also remains open whether the $2$-approximation for \MSMP{} is tight. And more generally, to close the gap between the upper bound and lower bound in Theorem~\ref{thm:+1gap}.

\subsubsection*{Acknowledgements}
The author thanks Bruce Shepherd for valuable discussions and suggestions that motivated some of this work.

\bibliography{REFERENCES}
\bibliographystyle{plain}

\appendix

\section{NP-Hardness of \MSMP}
\label{app:monotone-hardness}

In this section we argue that \MSMP{} is NP-Hard. This was already discussed in \cite{zhao2005greedy}, and we include it here for completeness.
Given a graph $G=(V,E)$ with a weight function $w:E \to \R_+$, let $\delta$ denote the weighted cut function $\delta (S) = \sum_{e \in E, \,\emptyset \subsetneq e \subsetneq S} w(e)$. Also, let $h:2^V \to \R_+$ denote the function $h(S)=\sum_{e\in E, \, e \subseteq S} w(e)$, which sums the weights of all the edges contained in the set $S$.
It is well-known and not hard to see that the function $f=\delta+h$ is non-negative monotone submodular. Moreover, the two optimization problems
$$
(\mbox{\KCUT}) \quad \min \; \frac{1}{2} \sum_{i=1}^k \delta(S_i): S_1 \uplus S_2 \uplus \cdots \uplus S_k = V \mbox{ and } S_i \neq \emptyset \mbox{ for all }i\in [k]
$$
and 
$$
\min \sum_{i=1}^k f(S_i): S_1 \uplus S_2 \uplus \cdots \uplus S_k = V \mbox{ and } S_i \neq \emptyset \mbox{ for all }i\in [k]
$$
are equivalent (in terms of the optimal solution), since for any feasible solution $S_1,S_2,\ldots,S_k$ we have 
$$\sum_{i=1}^k f(S_i) = \sum_{i=1}^k \delta(S_i) + \sum_{i=1}^k h(S_i) = w(E) + \frac{1}{2} \sum_{i=1}^k \delta(S_i) .$$
It then follows that an optimal solution to \MSMP{} is also an optimal solution to \KCUT{}. Hence, since \KCUT{} is NP-Hard when $k$ is part of the input, it follows that so is \MSMP{}.

It is worth pointing out that the APX-Hardness of \KCUT{} does not translate to \MSMP{}. This is because in order to get solutions of equal optimal values, the objective function should be $f':=\delta+h-\frac{w(E)}{k}$, so that for any feasible solution $S_1,S_2,\ldots,S_k$ we have  $\sum_{i=1}^k f'(S_i) = \frac{1}{2} \sum_{i=1}^k \delta(S_i)$. However, in this case the function $f'$ is not non-negative.

\end{document}